\documentclass[11pt]{article}
\usepackage{hyperref}
\usepackage{times}
\usepackage{float}
\usepackage{amsmath}
\usepackage{amsthm}
\usepackage{amstext}
\usepackage{amsmath}
\usepackage{amssymb}
\usepackage{amsfonts}
\usepackage{algorithm, algorithmicx, algpseudocode} 
\usepackage{graphicx}
\usepackage{color}
\usepackage{cite}
\usepackage{bbm}
\usepackage{centernot}
\usepackage{thmtools}
\usepackage{thm-restate}
\usepackage{enumitem} 
\usepackage{wrapfig}

\usepackage[margin=.96in]{geometry}

\newtheorem{lemx}{Lemma}
\newtheorem{obsx}{Observation}

\newif\ifusecompactlists
\usecompactlistsfalse

\ifusecompactlists
\usepackage{paralist}
\renewenvironment{itemize}[1][]{\begin{compactitem}}{\end{compactitem}}
\renewenvironment{enumerate}[1][]{\begin{compactenum}}{\end{compactenum}}
\fi

\newenvironment{proofsketch}{%
  \proof}{\endproof}

\usepackage{cleveref} 

\makeatletter
\def\@copyrightspace{\relax}
\makeatother

\renewcommand{\epsilon}{\varepsilon}

\newcommand{\conc}{\circ}
 \newcommand{\expect}{\mathbb{E} } 
\newcommand{\OPT}{\mbox{\sc OPT}}
\newcommand{\opt}{\OPT}

\newtheorem{theorem}{Theorem}
\newtheorem{lemma}[theorem]{Lemma}

\newtheorem{property}[theorem]{Property}
\newtheorem{observation}[theorem]{Observation}

\newtheorem{example}[theorem]{Example}
\newtheorem{claim}[theorem]{Claim}

\makeatletter
\newenvironment{proofof*}[1]{\par
  \pushQED{\qed}%
  \normalfont \topsep6\p@\@plus6\p@\relax
  \trivlist
  \item[\hskip\labelsep
        \bfseries
    #1\@addpunct{.}]\ignorespaces
}{%
  \popQED\endtrivlist\@endpefalse
}
\newenvironment{proofof}[1]{%
  \begin{proofof*}{Proof of #1}%
}{%
  \end{proofof*}\ignorespacesafterend
}
\renewcommand*{\@fnsymbol}[1]{\ensuremath{\ifcase#1\or\dagger\or
    \ddagger \or \mathsection \or \mathparagraph\else\@ctrerr\fi}}
\makeatother

\let\oldfootnoterule\footnoterule
\def\footnoterule{\vskip-1pt\oldfootnoterule\vskip1pt\relax}

\newcommand{\pparagraph}[1]{\vspace{0.1in}\noindent{\bf \boldmath #1.}}
       
\newcommand{\defn}[1]{\emph{\textbf{#1}}}
\renewcommand{\epsilon}{\varepsilon}

\author{
Michael A.~Bender\thanks{Department of Computer Science, 
Stony Brook University, 
Stony Brook, NY 11794-4400, USA. 
Email:~\texttt{\{bender,smccauley, shiksingh\}@cs.stonybrook.edu}.}
\and 
Samuel McCauley\footnotemark[1]
\and
Andrew McGregor\thanks{Department of Computer Science, 
University of Massachusetts, 
Amherst, MA 01003, USA.
Email:~\texttt{\{mcgregor,hvu\}@cs.umass.edu}.}
\and
Shikha Singh\footnotemark[1]
\and
Hoa T. Vu\footnotemark[3]
}

\date{}
\title{Run Generation Revisited:\\What Goes Up May or May Not Come Down}

\begin{document}
\maketitle


\begin{abstract}
In this paper, we revisit the classic problem of run generation.  Run
generation is the first phase of external-memory sorting, where the objective
is to scan through the data, reorder elements using a small buffer of size $M$,
and output \defn{runs} (contiguously sorted chunks of elements)  that are as
long as possible.

We develop algorithms for minimizing the total number of runs (or equivalently,
maximizing the average run length) when the runs are allowed to be sorted or
reverse sorted.  We study the problem in the online setting, both with and
without resource augmentation,  and in the offline setting. 

\begin{itemize}

\item We analyze alternating-up-down replacement selection (runs alternate
between sorted and reverse sorted),  which was studied by Knuth as far back as
1963.  We show that this simple policy is asymptotically optimal.
Specifically, we show that alternating-up-down replacement selection is
2-competitive and no deterministic online algorithm can perform better. 

\item We give online algorithms having smaller competitive ratios with resource
augmentation.  Specifically, we exhibit a deterministic algorithm that, when
given a buffer of size $4M$, is able to match or beat any optimal algorithm
having a buffer of size $M$. Furthermore, we present a randomized online
algorithm which is $7/4$-competitive when given a buffer twice that of the
optimal.

\item We demonstrate that performance can also be improved with a small amount
of foresight.  We give an algorithm, which is $3/2$-competitive, with
foreknowledge of the next $3M$ elements of the input stream.  For the extreme
case where all future elements are known, we design a PTAS for computing the
optimal strategy a run generation algorithm must follow.

\item We present algorithms tailored for ``nearly sorted'' inputs which are
guaranteed to have optimal solutions with sufficiently long runs.
 
\end{itemize}
\end{abstract}

\thispagestyle{empty} 
\newpage 
\setcounter{page}{1}

\section{Introduction}
\label{sec:intro}

External-memory sorting algorithms are 
 tailored for data sets too large to fit in main memory.
Generally,
these algorithms begin their sort by bringing chunks of data into main memory, 
sorting within
memory, and writing back out to disk in sorted sequences, called
\defn{runs}~\cite{Martinez-PalauDoLa10, Knuth98, Estivill-CastroWo94, Gassner67}.

We revisit the classic problem of how to maximize the length of these
runs, the \defn{run-generation problem}.
The run-generation problem  has been studied in its various guises for over 50 years~\cite{Martinez-PalauDoLa10, Knuth63, Lin93, LinLa97, FrazerWo72, Friend56, Gassner67,Espelid76}.

The most well-known external-memory sorting algorithm is multi-way 
merge sort~\cite{LarsonGr98, Estivill-CastroWo94, BittonDe83, ZhengLa96, Graefe06, Larson03, ZhangLa97, AggarwalVi88, TingWa77}. 
The multi-way merge sort is formalized in the 
\defn{disk-access machine}\footnote{The external-memory  model, 
also called the I/O model, applies to any two levels of the 
memory hierarchy.}
(\defn{DAM}) model
of  Aggarwal and Vitter~\cite{AggarwalVi88}.
If $M$ is the size of RAM and data is transferred between main memory and disk in blocks of size $B$, then an $M/B$-way merge sort has a complexity of 
$O\big(({N}/{B}) \log_{M/B} {({N}/{B})}\big)$ I/Os, where $N$ is the number of elements to be sorted. This is the best possible \cite{AggarwalVi88}.

A top-down description of multi-way merge sort  follows. 
Divide the input into $M/B$ subproblems, recursively sort each subproblem, and merge them together in one final scan through the input.
The base case is reached when each subproblem has size $O(M)$, and therefore fits into RAM.

A bottom-up description of the algorithm starts with the base case, which is the 
run-generation phase.
Na\"i{}vely, we can always generate runs of length $M$:
ingest $M$ elements into memory, sort them, write them to disk, and then repeat. 

The point of run generation is to produce runs \emph{longer} than $M$. 
After all, with typical values of $N$ and $M$, 
we rarely need more than 
one or two passes over the data after the initial run-generation phase.  
Longer runs can mean fewer passes over the data or less memory consumption during the merge phase of the sort. 
Because there are few scans to begin with, even if we only do one fewer scan, the cost of a merge sort is decreased by a significant percentage. 
Run generation has further advantages in databases even when a full sort is not required \cite{Graefe06,KellerGrMa91}.

\pparagraph{Replacement Selection} 
The classic algorithm for run generation is called \defn{replacement selection}
 \cite{Goetz63,Knuth98, Larson03}. 
We describe replacement selection below by assuming that the elements can be read into memory and written to disk 
\emph{one at a time}.

To create an increasing run starting from an initially full internal memory, proceed as follows:

\begin{enumerate}[noitemsep,nolistsep,leftmargin=*]
\item From main memory, select the smallest element\footnote{Observe that  data structures such as in-memory
heaps can be used to identify the smallest elements in memory.  However,  from
the perspective of minimizing I/Os, this does not matter---computation is free
in the DAM model.} at least as large as every element 
in the current run.
\item If no such element exists, then the run ends; select the smallest element in the buffer.
\item \defn{Eject} that element, and \defn{ingest} the next element, so that the memory stays full.
\end{enumerate}

Replacement selection can deal with input elements one at a time, even though the DAM model transfers input 
between RAM and disk $B$ elements at a time. 
To see why, consider two additional blocks in memory, 
an ``input block,'' which stores elements recently read  from disk, and an 
``output block,'' which 
stores elements that have already been placed in a run and will be written back
to disk.   To ingest, take an element from the input block, and to eject
an element, put the element in the output block. 
When the input block becomes empty, fill it  from disk
and when the output block fills up, flush it to disk. 
Similar to previous work, in this paper, we ignore these two blocks.

\pparagraph{Properties of Replacement Selection}
It has been known for decades that when the input appears in random order, then 
the expected length of a run is actually $2M$, not $M$~\cite{Friend56, Gassner67, Knuth63}. 
In~\cite{Knuth98}, Knuth gives memorable intuition about this result, conceptualizing the buffer as a snowplow traveling along a circular track.

Replacement selection performs particularly well on nearly sorted data
(for many intuitive notions of ``nearly''), and 
the runs generated are much larger than $M$. For example, when each element in the input appears at a distance at most $M$ from its actual rank, replacement selection produces a single run.

On the other hand, replacement selection performs poorly on reverse-sorted data.
It produces runs of length $M$, which is the worst possible. 

\pparagraph{Up-Down Replacement Selection}
From the perspective of the sorting algorithm, it 
matters little, or not at all,
whether the initially generated runs are sorted or reverse sorted.

This observation has motivated researchers to think about run generation when
the replacement-selection algorithm has a choice about whether to generate an
\defn{up run} or a \defn{down run}, each time a new run begins.

Knuth~\cite{Knuth63} analyzes the performance of replacement selection that
alternates deterministically between generating up runs and down runs. He shows
that for randomly generated data, this alternative policy performs \emph{worse},
generating runs of expected length $3M/2$, instead of $2M$.

Martinez-Palau et al.~\cite{Martinez-PalauDoLa10} revive this idea in an experimental 
study. Their two-way-replacement-selection  algorithms 
 heuristically choose
between whether the run generation should go up or down.
Their experiments find that  two-way replacement selection (1) is slightly worse
than replacement selection for random input (in accordance with
Knuth~\cite{Knuth63}) and (2)  produces 
significantly longer runs on
inputs that have mixed up-down runs and reverse-sorted inputs.

\pparagraph{Our Contributions}
The results in our present paper complement these earlier results. 
In contrast to Knuth's negative result for random  inputs~\cite{Knuth63}, 
we show that strict up-down alternation is best possible for worst-case inputs. 
Moreover, we give better competitive ratios with resource augmentation and lookahead, 
which helps explain  why heuristically choosing between up and down runs based on what is currently in memory may lead to better solutions. 
Resource augmentation is a standard tool used in
competitive analysis~\cite{EnglertWe05, ChanMeSi11,sleator1985amortized,sun2013improved,epstein2007online,chekuri2004multi}
to empower an online algorithm when comparing against an omniscient and all-powerful optimal algorithm.

Up-down run generation boils down to figuring out, each time a run ends, whether the next run should be an up run or a down run. 
The objective is to minimize the number of runs output.\footnote{Note that for a given input,
minimizing the number of runs output is equivalent to maximizing the average length of runs output.}
We establish the following: 

\smallskip

\begin{enumerate}[noitemsep,nolistsep,leftmargin=*]
\item \emph{Analysis of  alternating-up-down replacement selection.} We revisit
(online) alternating-up-down replacement selection, which was earlier analyzed
by Knuth \cite{Knuth63}.  We prove that alternating-up-down replacement selection is 2-competitive and
asymptotically optimal for deterministic algorithms. To put this result in context, it is known  that  up-only replacement selection
is a constant factor better than up-down replacement selection for random inputs,  but  can be
an unbounded factor worse than optimal for arbitrary inputs.

\item \emph{Resource augmentation with extra buffer.}
We analyze the effect of augmenting the buffer available to an online algorithm on its performance. 
We show that with a constant factor larger buffer, it is possible to perform better than twice optimal.
Specifically, we exhibit a deterministic algorithm that, when given a
buffer of size $4M$, matches or beats any optimal algorithm having a buffer of size $M$. 
We also design a randomized online algorithm which is $7/4$-competitive using a $2M$-size buffer.

\item \emph{Resource augmentation with extra visibility.} We show that performance factors 
can also be improved, without augmenting the buffer, if an algorithm has limited foreknowledge of the input.
In particular, we propose a deterministic algorithm which attains a competitive ratio of $3/2$, using its regular buffer of
size $M$, with a {\em lookahead} of $3M$ incoming elements of the input (at each step).

\item \emph{Better bounds for nearly sorted data.}
We give algorithms that perform well on inputs that have some inherent
sortedness.
We show that the greedy offline algorithm is optimal for inputs on which the optimal runs are at least $5M$ elements long.
We also give a $3/2$-competitive algorithm with $2M$-size buffer when the optimal runs are at least $3M$ long.
These results are reminiscent of previous literature studying sorting on inputs with 
``bounded disorder''~\cite{ChandramouliGo14} and adaptive sorting
algorithms~\cite{Mallows63,wiki:timsort,Estivill-CastroVlWo92}.

\item \emph{PTAS  for the offline problem.}
We give a polynomial-time approximation scheme for the offline run-generation problem. 
Specifically, our offline polynomial-time approximation
algorithm  guarantees a $(1 + \epsilon)$-approximation to the optimal solution.
We first give an algorithm with the running time of $O(2^{1/\epsilon}N\log N)$ and then improve the running time 
to $O\left( \big( \frac{1+\sqrt{5}}{2} \big)^{1/\epsilon}N\log N\right)$.

\end{enumerate}

\pparagraph{Paper Outline}
The paper is organized as follows. In \Cref{sec:up_down_run_generation}, we
formalize the up-down run generation problem and provide necessary notation.
\Cref{sec:structural_properties} contains important structural properties of
run generation and key lemmas used in analyzing our algorithms.  
Analysis of alternating-up-down replacement selection and online lower bounds are in
\Cref{sec:online}. Algorithms with resource augmentation, along with properties of the greedy algorithm,
are presented in \Cref{sec:res-aug}. 
The offline version of the problem is studied in \Cref{sec:offline}.  Improvements on
well-sorted input are presented in \Cref{sec:nearly-sorted}.
\Cref{sec:related} summarizes related work and we conclude with open problems in \Cref{sec:conclusion}. 
Due to space constraints, we defer some proofs to the appendix (\Cref{sec:omittedproofs}).


\section{Up-Down Run Generation}\label{sec:up_down_run_generation}

In this section, we formalize the up-down run generation problem and introduce notation.

\subsection{Problem Definition} 
An instance of the up-down run generation problem is a stream $I$ of $N$ elements. The
elements of $I$ are presented to the algorithm one by one, in order. They can
be stored in the memory of size $M$ available to the algorithm, which we henceforth refer to as the \defn{buffer}.  
Each element occupies one slot of the
buffer. 
In general, the model allows duplicate elements,
although some results, particularly in \Cref{sec:resource_augmentation} and \Cref{sec:nearly-sorted}, do require
uniqueness.

We say that an algorithm $A$ \defn{reads} an element of $I$ when $A$ transfers 
the element from the input sequence to the buffer.
We say that an algorithm  $A$ \defn{writes}
an element
when $A$ ejects the element from its buffer and appends it to the \defn{output sequence} $S$. 

Every time an element is written, its slot in the buffer 
becomes free. Unless stated otherwise, the next element from
the input takes up the freed slot. Thus the buffer is always full,
except when the end of the input is
reached and there are fewer than $M$ unwritten elements.\footnote{Reading in the next element of the input
when there is a free slot in the buffer never hurts the performance of any algorithm.
However, we allow the algorithm in the proof of \Cref{lem:pruning} to maintain free slots in the buffer 
to simplify the analysis.}

An algorithm can decide which element to eject from its buffer based on (a)~the current
contents of the buffer and (b)~the last element written.
The algorithm may
also use $o(M)$ additional words to maintain its internal state (for example, it can store the direction of
the current run). However, the algorithm cannot arbitrarily access $S$ or
$I$---it can only append elements to $S$, and access the next in-order element
of $I$.
We say the algorithm is at \defn{time step} $t$ if it has written exactly $t$ elements. 

A \defn{run} is a sequence of sorted or reverse-sorted elements.  
The cost of the algorithm is the smallest number of runs we can use to partition its output.  
Specifically, 
    the number of runs in an output $S$, denoted $R(S)$, is the smallest 
    number of mutually disjoint sequences $S_1,S_2,\ldots,S_{R(S)}$ such that
        each $S_i$ is a run and 
        $S = S_1 \conc \cdots \conc S_{R(S)}$ where $\conc$ indicates concatenation. 

We let $\OPT(I)$ be the minimum number of runs of any possible output sequence on input $I$,
i.e., the number of runs generated by the optimal offline algorithm.
If $I$ is clear from context, we denote this as $\OPT$.  Our goal is to
give algorithms that perform well compared to $\OPT$ for every $I$. We say that an online algorithm is \defn{$\beta$-competitive} if on any input, its output $S$ satisfies $R(S) \leq \beta \opt$. 

At any time step, an algorithm's \defn{unwritten-element sequence} is 
comprised of the contents of the
buffer, concatenated with the remaining (not yet ingested) input elements.  
For the purpose of this definition,
we assume that the elements in the buffer are stored in their arrival order
(their order in the input sequence $I$).  

Time step $t$ is a \defn{decision point} or \defn{decision time step} for an algorithm $A$ if $t=0$ or if $A$ finished writing a run at $t$.  At a decision point, $A$ needs to decide whether the next run will be increasing or decreasing.

\subsection{Notation}

We employ the following notation. 
We use $(x\nearrow y)$ to denote the increasing sequence $x, x+1, x+2, \ldots,
y$ and $(x\searrow y)$ to denote the decreasing sequence $x, x-1, x-2 \ldots,
y$.  We use $\conc$ to denote concatenation: if $A = a_1,a_2,\ldots, a_k$ and
$B = b_1,b_2,\ldots, b_{\ell}$ then $A\conc B = a_1,a_2,\ldots a_k, b_1, b_2,
\ldots, b_{\ell}$.

Let $A = a_1,a_2,\ldots,a_k$.
We use  
$A \oplus x$ to denote the sequence $a_1 + x, a_2 + x,\ldots a_k + x$. 
Similarly, we use $A \otimes x$ to denote  the sequence 
$a_1x, a_2x,\ldots, a_kx$.

Let $A,B$ be sequences. We say $A$ \defn{covers} $B$ if for all $e \in B, e \in
A$.
A \defn{subsequence} of a sequence $A = a_1, \ldots, a_k$ is a sequence $B =
a_{n_1}, a_{n_2}, \ldots, a_{n_{\ell}}$ where  $1\leq n_1 < n_2 <\ldots < n_{\ell}\leq k$.

\section{Structural Properties} 
\label{sec:structural_properties}
In this section, we identify structural properties of the problem and
the tools used in the analysis of our algorithms, which will be important in the rest of the paper.  

\subsection{Maximal Runs}

We show that in run generation, it is never a good idea to end a run early, and never a good idea to ``skip over'' an element (keeping it in buffer instead of writing it out 
as part of the current run).

To begin, we show that adding elements to an input sequence never decreases the number of runs.
Note that if $S'$ is a subsequence of $S$, then $R(S')\leq R(S)$ by definition.

\begin{lemma}
\label{lem:triangle}
Consider two input streams $I$ and $I'$.  If $I'$ is a subsequence of $I$, then 
$\OPT({I'}) \leq \OPT(I)$.
\end{lemma}

\begin{proof} 
Let $A$ be an algorithm with input stream $I$ and output $S$.
Suppose that  $A$ produces the optimal number of runs on $I$, that is $R(S) =
\OPT(I)$. Consider an algorithm $A'$ on $I'$.  Algorithm $A'$ performs the same operations as $A$, but when it reaches an element that is not in $I'$ (but is in $I$), it executes a no-op.  These no-ops mean that the buffer of $A'$ may not be completely full, since elements that $A$ has in buffer do not exist in the buffer of $A'$.  Let $S'$ be the output of $A'$; $S'$ is a subsequence of $S$.
 
    Then $\OPT({I'}) \leq R(S') \leq R(S) = \OPT(I)$.
\end{proof}

A \defn{maximal increasing run} is a run generated using the following rules (a \defn{maximal decreasing run} is defined similarly):
\begin{enumerate}[noitemsep,nolistsep,leftmargin=*]
\item Start with the smallest element in the buffer and always write the smallest element that is larger than the last element written. 
\item End the run only when no element in the buffer can continue the run, i.e., all elements in buffer are smaller than the last element written.
\end{enumerate}

\begin{lemma} \label{lem:maximalcover}
    At any decision time step, a maximal increasing (decreasing) run $r$ covers every other (non-maximal) increasing (decreasing) run $r'$. 
\end{lemma}

A \defn{proper algorithm} is an algorithm that always writes maximal runs. We say an output is proper if it is generated by a proper algorithm.
We show that there always exists an optimal proper algorithm. 

\begin{theorem}
    \label{thm:maximal}
    For any input $I$, there exists a proper algorithm $A$ with output $S$ such that $R(S) = \OPT(I)$.
\end{theorem}
\begin{proof}
    We prove this by induction on the number of runs. If there is only one run, it must be maximal. Assume that all inputs $I_t$ with $\OPT(I_t) = t$ 
have a maximal proper algorithm.  Consider an input $I_{t+1}$ with $\OPT(I_{t+1}) = t+1$.  Assume that an optimal algorithm on $I_{t+1}$ is $A_O$, and it is not proper; we will construct a proper $A$ with the same number of runs.  The first run $A$ writes is maximal and has the same direction of the first run that $A_O$ writes; the first run $A_O$ writes may or may not be maximal.  Then $A$ is left with an unwritten-element sequence $I_A$ and $A_O$ is left with $I_O$.  Note that $\OPT(I_O) = t$ by definition.
    
    By \Cref{lem:maximalcover}, $I_O$ is a subsequence of $I_A$.  Then by \Cref{lem:triangle}, $\OPT(I_A)\leq\OPT(I_O)$.  Then by the inductive hypothesis, $I_A$ has an optimal proper algorithm.  Thus $A$ is a proper algorithm generating the optimal number of runs.
\end{proof}

In conclusion, we have established that it always makes sense for an algorithm to write maximal runs. Furthermore, we use the following property of proper algorithms throughout the rest of the paper.

\begin{property}  \label{pro:wlog}
    Any proper algorithm satisfies the following two properties:
    \begin{enumerate}[noitemsep,nolistsep]
        \item At each decision point, the elements of the buffer must have arrived while the previous run was being written. 
        \item A new element can not be included in the current run if the  element just written out is larger (smaller) and the current run is increasing (respectively, decreasing). 
    \end{enumerate}
\end{property}

\subsection{Analysis Toolbox}

We now present observations and lemmas that play an integral role
in analysing the algorithms presented in the rest of the paper.

\begin{observation}
\label{obs:remaindersequence} Consider algorithms $A_1$ and $A_2$ on input $I$.
Suppose that at time step $t_1$ algorithm $A_1$ has written out all the
elements that algorithm $A_2$ already wrote out by some previous time step
$t_2$.  Then, the unwritten-element sequence of algorithm $A_1$ at time step
$t_1$ forms a subsequence of  the unwritten-element sequence of algorithm $A_2$
at time step $t_2$.  
\end{observation}

\begin{lemma} 
    \label{lem:remaindersequence2} 
    Consider a proper algorithm $A$.
At some decision time step, $A$ can write $k$ runs $p_1 \conc \cdots \conc p_k$
or $\ell$ runs $q_1 \conc \cdots \conc q_\ell$ such that $|p_1 \conc \cdots
\conc p_k|  \geq |q_1 \conc \cdots\conc q_\ell|$. Then $p_1 \conc \cdots \conc
p_k \conc p_{k+1}$, where $p_{k+1}$  is either an up or down run, covers $q_1
\conc \cdots \conc q_\ell$. 

Therefore, the unwritten-element sequence after $A$ writes $p_{k+1}$ (if $A$
writes $p_1\conc\cdots\conc p_{k+1}$) is a subsequence of the
unwritten-element sequence after $A$ writes $q_{\ell}$ (if $A$ writes
$q_1\conc\cdots\conc q_{\ell}$).  \end{lemma}

\begin{proof} Since $|p_1 \conc \cdots \conc p_k|  \geq |q_1 \conc \cdots\conc
q_\ell|$, the set of elements that are in $q_1 \conc \cdots\conc q_\ell$ but
not in $p_1 \conc \cdots \conc p_k$ have to be in the buffer when $p_k$ ends.
By \Cref{pro:wlog}, $p_{k+1}$ will write all such elements.  \end{proof}

The next theorem serves as a template for analyzing the algorithms in this paper.
It helps us restrict our attention to comparing the output of our algorithm against that of the optimal
in small {\em partitions}. 
We show that if in every partition $i$, an algorithm writes $x_i$ runs that cover the first $y_i$
runs of an optimal output (on the current unwritten-element sequence),
and $x_i/y_i \leq \beta$,
then the algorithm outputs no more than $\beta \opt$ runs.

\begin{theorem}
\label{thm:template}
Let $A$ be an algorithm with output $S$.  Partition $S$ into 
$k$ contiguous subsequences $S_1, S_2\ldots S_k$. Let $x_i$ be the number of runs in $S_i$.
For $1 < i \leq k$, let $I_i$ be the unwritten-element sequence after $A$
outputs $S_{i-1}$; let $I_1 = I$ and $I_{k+1} = \emptyset$.  
Let $\alpha,\beta \geq 1$. For each $I_i$,
let $S'_i$ be the output of an optimal algorithm on $I_i$. 

If for all $i \leq k$,  
$S_i$ covers the first $y_i$ runs of $S'_i$, and
$x_i/y_i \leq \beta$,
then $R(S)\leq \beta \OPT$.
Similarly, 
if for all $i \leq k$,  
$S_i$ covers the first $y_i$ runs of $S'_i$, and
$\expect [x_i]/y_i \leq \alpha$,
then $\expect [R(S)]\leq \alpha \OPT$.
\end{theorem}

\begin{proof}
    Consider $I'_i$, the unwritten element sequence at the end of the first $y$
runs of $S'_{i-1}$ (we let $I'_1 = I$).  
    We show that $\OPT(I_i) \leq \OPT - \sum_{j=1}^{i-1}y_i$ for all $1 \leq i \leq
k$ using induction.
    Note that $\OPT(I_1) = \OPT$ (the base case).  Induction hypothesis: assume $\OPT(I_i) \leq \OPT - \sum_{j=1}^{i-1}y_i$.
    Since $S_{i+1}$ covers the first $y$ runs of $S_{i+1}'$, by \Cref{obs:remaindersequence}, $I_{i+1}$ is a
subsequence of $I'_{i+1}$. Then by \Cref{lem:triangle}, $\OPT(I_{i+1})\leq
\OPT(I'_{i+1})$.  By definition, for $i > 1$, 
\[\OPT(I'_{i+1}) = \OPT(I_{i}) - y_i \leq \OPT - \sum_{j=1}^i y_i \ .\]  
Therefore, $\OPT(I_{i+1}) \leq  \OPT - \sum_{j=1}^i y_i $. When $i = k$, we have $\OPT(I_{k+1}) 
 \leq \OPT -  \sum_{j=1}^k y_i$. But since $I_{k+1}$
contains no elements, $\OPT(I_{k+1}) = 0$, and we have $ \sum_{j=1}^k y_i \leq \OPT$.  
Since $R(S)=  \sum_{j=1}^k x_i$, and $\sum_{i=1}^k x_i \leq \beta \sum_{i=1}^k y_i$,
we have the following:\\
\[
R(S) = \frac{\sum_{i=1}^k x_i} {\opt} \cdot \opt \leq \frac{\sum_{i=1}^k x_i} {\sum_{i=1}^k y_i} \cdot \opt \leq \beta \opt.
\]

We also have the same in expectation, that is,
\[
\expect [R(S)] = \expect[\sum_{i=1}^n x_i ] \leq \alpha \sum_{i=1}^n y_i \leq \alpha \cdot \opt. 
\]
\end{proof}

\section{Up-Down Replacement Selection}
\label{sec:online}

We begin by analyzing the {\em alternating up-down replacement selection}, which deterministically alternates between writing (maximal)
up and down runs. Knuth \cite{Knuth63} showed that when the input elements arrive in a random order (all
permutations of the input are equally likely), alternating-up-down replacement selection
performs worse than standard replacement selection (all up runs). Specifically, he showed
that the expected length of runs generated by up-down-replacement selection is $1.5M$ on random input,
compared to the expected length of $2M$ of replacement selection.

In this section, we show that for deterministic online algorithms, alternating-up-down replacement selection is, in fact,  
asymptotically optimal for {\em any} input. It generates at most twice the optimal number of runs in the worst case. 
This is the best possible---no deterministic algorithm can have a better competitive ratio.


\subsection{Alternating-Up-Down Replacement Selection is 2-competitive}
\label{sec:two-approx}


We begin by giving a structural lemma, analyzing identical runs on two inputs
in which one input is a subsequence of the other.

\begin{lemma} 
    \label{lem:subsequence}
Consider two inputs $I_1$ and $I_2$, where $I_2$ is a subsequence of $I_1$. 
Let $S_1$ and $S_2$ be proper outputs of $I_1$ and $I_2$ such that:
\begin{enumerate}[noitemsep,nolistsep]
    \item $S_1$ and $S_2$ have initial runs $r_1$ and $r_2$ respectively,
    \item $r_1$ and $r_2$ have the same direction 
\end{enumerate}
Let the unwritten-element sequence after $r_1$ and $r_2$ 
be $I_1'$ and $I_2'$ respectively. 
Then 
$I_2'$ is a subsequence of $I_1'$.
\end{lemma}

\begin{proof} 
    Assume that $r_1$ and $r_2$ are up runs (a similar analysis works for down
runs).  Let $r_2'$ be a run that is a subsequence of $r_1$, consisting of all
elements of $r_1$ that are also in $I_2$.  Then $r_2'$ can be produced by an
algorithm $A'$ that mirrors the algorithm $A$ that generates $r_1$.  When $A$
reads or writes an element in $I_2$, $A'$ reads or writes that element; when
$A$ reads or writes an element not in $I_2$, $A'$ does nothing. 
    Since $r_2$ is maximal, it covers $r_2'$ by \Cref{lem:maximalcover}.
\end{proof}

\begin{theorem} 
\label{thm:2competitive}
    Alternating up-down replacement selection is 2-competitive.
\end{theorem}
\begin{proof}
    We show that we can apply \Cref{thm:template} to this algorithm with $\beta = 2$. 
    
    In any partition that is not the last one of the output, the alternating algorithm writes a maximal up
run $r_u$ and then writes a maximal down run $r_d$.  We must show that $r_u\conc r_d$ covers any
run $r_O$ written by a proper optimal algorithm on $I_r$, the unwritten element
sequence at the beginning of the partition. 

    If $r_O$ is an up run, then $r_O = r_u$ and thus is covered by $r_u \conc
r_d$. If $r_O$ is a down run, consider $I'$, the unwritten-element sequence after
$r_u$ is written; $I'$ is a subsequence of $I_r$.  By
\Cref{lem:subsequence} (with $I_1 = I_r$ and $I_2 = I'$), $r_u\conc r_d$ covers
$r_O$.

	In the last partition, the algorithm can write at most two runs while any optimal output must contain at least one run. Hence $x_i/y_i \leq 2$ in all partitions as required.
\end{proof}

\subsection{Lower Bounds on Online Algorithms for Up-Down Run Generation}
\label{sec:lowerbounds}

Now, we show that no deterministic online algorithm can hope to perform better than alternating-up-down replacement selection. 
Then, we partially answer the question of whether randomization helps overcome this impossibility result. 
Specifically, we show that no randomized algorithm
can achieve a competitive ratio better than $3/2$. We provide the main ideas of
the proofs here and defer the details to \Cref{sec:omittedproofs}. 

\begin{theorem}
    \label{thm:2lower}
Let $A$ be any online deterministic algorithm with output $S_I$ on input $I$. Then there are arbitrarily long $I$ such that
$R(S_I) \geq 2 \OPT(I)$. 
\end{theorem}

\begin{proofsketch}
Given any $M$ elements in the buffer, every time $A$ commits to a run direction
(up/down), the adversary sets the incoming elements such that they do not help the current run. 
Thus, $A$ is forced to have runs of length at most $M$ while $\OPT$ (since it has knowledge of the future) can do better. 
\end{proofsketch}

We also give a lower bound for randomized algorithms using similar ideas; however, in this case we do not have a matching upper bound.
We use Yao's minimax principle to prove this bound.
That is, we generate a randomized input and show that 
any deterministic algorithm cannot perform better than $3/2$ times $\OPT$ on that input against an oblivious adversary.  

\begin{theorem}\label{thm:randlower} 
    Let $A$ be any online, randomized algorithm.  
    Then there are arbitrarily long input sequences such that $\expect [R(S_I)] \geq (3/2)\OPT(I)$.
\end{theorem}

\section{Run Generation with Resource Augmentation}\label{sec:res-aug}
\label{sec:resource_augmentation}

In this section, we use resource augmentation to circumvent the impossibility
result on the performance of deterministic online algorithms. We consider two
kinds of augmentation: \begin{itemize}[noitemsep, nolistsep] \item {\em Extra
Buffer:} The algorithm's buffer is actually a constant factor larger, that is,
it can use its large buffer to read elements from the input, rearrange them,
and write to the output.  \item {\em Extra Visibility:} The algorithm's buffer
is restricted to be of size $M$ but it has prescience---the algorithm can {\em
see} some elements in the immediate future (say,  the next $3M$ elements),
without the ability to write them early.  \end{itemize} We present algorithms
that, under the above conditions, achieve a competitive ratio better than $2$
when compared against an optimal offline algorithm with a buffer of size $M$. 

Resource augmentation is a common tool used in competitive
analysis~\cite{EnglertWe05,
ChanMeSi11,sleator1985amortized,sun2013improved,epstein2007online,chekuri2004multi}.
It gives the online algorithm power to make better decisions and exclude worst
case inputs, allowing us to compare the performance, more realistically,
against an all-powerful offline optimal algorithm.

The results in this section require the elements of the input to be unique.
Duplicate elements can nullify the extra ability to see or write future
(non-repeated) elements which is provided by visibility and buffer-augmentation
respectively.  For example, consider the input, \vspace{-1ex} $$I = (99,
101,\underbrace{100,\ldots,100}_{cM-2 \mbox { times}},\ldots).$$ \vspace{-2ex}  

On input $I$, any algorithm with $cM$-size buffer or visibility is as powerless
as the one without any augmentation.

Note that the assumption of distinct elements in run generation is not new.
Knuth's analysis of the average run lengths \cite{Knuth63} also requires
uniqueness.

We begin by analyzing the \defn{greedy algorithm} for run generation. Greedy is
a proper algorithm which looks into the future at each decision point,
determines the length of the next up and down run and writes the longer run. 

Greedy is not an online algorithm. However, it is central to our resource
augmentation results. The idea of resource augmentation, in part, is that the
algorithm can use the extra buffer or visibility to determine, at each decision
point, which direction (up or down) leads to the longer next run. 

We next look at some guarantees on the length of a run chosen by greedy (or the
{\em greedy run}) and also on the run that is not chosen by greedy (or the {\em
non-greedy run}).

   \subsection{Greedy is Good but not Great}\label{sec:greedy}

We first show that greedy is not optimal. The following example demonstrates that greedy can be a factor of $3/2$ away from optimal.  

\begin{example} \label{example:greedy-vs-optimal}
Consider the input
$I =   I_1 \conc (I_1 \oplus 10M ) \conc (I_1 \oplus 20M) \conc \cdots \conc (I_1 \oplus 10cM)$, where
\begin{flalign*}
I_1& =  (4M+4 \nearrow 5M+3) \conc (M+2) \conc (5M+4 \nearrow 6M+3)& \\ 
      & \conc (2M+1 \nearrow 3M-1)  \conc (4M+3 \searrow 3M+4)  \conc (2M \searrow M+3) \conc (M+1 \searrow 1).&
\end{flalign*}
\end{example}

On input $I$ above, writing down runs repeatedly produces $2c$ runs; two for each $I \oplus i10M$. 
On the other hand, the output of greedy is 
$S_1 \circ (S_1 \oplus 10M ) \conc \cdots \conc (S_1 \oplus c10M )$, where  $S_1 = (4M+4
\nearrow 6M+3) \conc (M+2) \conc (2M+1 \nearrow 3M-1) \conc (3M+4 \nearrow
4M+3) \conc (2M \searrow M+3) \conc (M+1 \searrow 1)$ which contains $3c$ runs.  

Next, we show that all the runs written by the greedy
algorithm (except the last two) are guaranteed to have length at least $5M/4$. In contrast, 
up-down replacement selection can have have runs of length $M$ in the worst case.

\begin{theorem}\label{thm:greedy_gaunratee}
Each greedy run, except the last two runs, has length at least $M + \lceil\lfloor M/2\rfloor/2\rceil$. 
\end{theorem}

We now bound how far into the future an algorithm must see to be able to determine
which direction greedy would pick at a particular decision point. Intuitively, an algorithm should never have to choose between a very long up run and a very long down run. We formalize this idea about the non-greedy run not being too long in the following lemma. 

\begin{lemma} \label{lem:3m}
Given an input $I$ with no duplicate elements. Let the two possible initial increasing and decreasing runs be $r_1$ and $r_2$. Then $|r_1| < 3M$ or $|r_2| < 3M$.  
\end{lemma}

The next example shows that the above bound is tight.
\begin{example}
Consider the input $I = I_1 \conc I_2 \conc I_3$, where
\vspace{-\abovedisplayskip}
\begin{align*}
    I_1 =& (1 \nearrow (M-1)) \otimes M, \mbox{~~} I_2 = (M^2 \searrow M^2 - M+1) \\[-3pt]
I_3 =& (M-1 \searrow 1) \conc (M^2+2 \nearrow M^2 + M + 1) \ .\end{align*} 

     \vspace{-\belowdisplayskip}
     \noindent
     Then, 
     \vspace{-\abovedisplayskip}
 \begin{align*}
 r_1 =& ((1 \nearrow (M-1)) \otimes M) \conc (M^2 - M +1 \nearrow M^2 + M +1) \\[-3pt]
 r_2 =& (M^2 \searrow M^2 - M+1) \conc ((M-1 \searrow 1) \otimes M) \conc (M-1 \searrow 1).
 \end{align*}  

     \vspace{-\belowdisplayskip}
     Thus, we have $|r_1|= 3M$ and $|r_2| = 3M-1$.
\end{example}

The following lemma sheds some light on the choices made by an optimal algorithm with respect to that of greedy. 
It says, roughly, that if at any decision point, an optimal algorithm chooses to write the non-greedy run, 
and then writes the next run in the opposite direction, it performs
no better than an optimal algorithm which chooses the greedy run in the first place.

\begin{lemma} \label{lem:pruning}
At any decision time step
consider two possible next maximal runs $r_1$ and $r_2$. If $|r_1| \geq |r_2|$, then one of the following is the prefix of an optimal output on the unwritten-element sequence: 
\begin{enumerate}[noitemsep,nolistsep]
\item $r_1 \circ r_3$ where $r_3$ is a maximal run after $r_1$ and it can be either up or down. 
\item $r_2 \circ r_4$ where $r_4$ is maximal run after $r_2$ with the same direction of $r_2$.
\end{enumerate}
\end{lemma}




   \subsection{Online Algorithms with Resource Augmentation}
We now present several online algorithms which use resource augmentation (buffer or visibility)
to determine an up-down replacement selection strategy, beating the competitive ratio of $2$.
For a concise summary of results, see \Cref{fig:results}.

\pparagraph{Matching OPT using $4M$-size Buffer} 
We present an algorithm with $4M$-size buffer that writes no more runs than an optimal algorithm with an $M$-size buffer.
Later on, we prove that $(4M-2)$-size is necessary even to be $3/2$-competitive; thus this augmentation result is optimal up to a constant.

Consider the following deterministic algorithm with a $4M$-size
buffer. 
The algorithm reads elements until its buffer is full.
It then uses the contents of its buffer to 
determine, for an algorithm
with buffer size $M$, if the maximal up run or the maximal
down run would be longer. \label{step:greedy}
If the maximal up run is longer,
the algorithm uses its full buffer (of size $4M$) to write a
maximal up run; otherwise it writes a maximal down run.
The algorithm stops when there is no element left to write.

\begin{theorem}
    \label{thm:optaug} 
    Let $A$ be the algorithm with a $4M$-size buffer described above.  On any input $I$, $A$ never writes more runs than an optimal algorithm with buffer size $M$.
\end{theorem}

\begin{proofsketch}
At each decision point, $A$ determines the direction that a greedy algorithm on the same unwritten element sequence, 
but with a buffer of size $M$, would have picked. It is able to do so using its $4M$-size buffer because,
by \Cref{lem:3m}, we know the length of the non-greedy run is bounded by $3M$.
Note that it does not need to write any elements during this step.
In each partition, $A$ writes a maximal run $r$ in the greedy direction and thus covers the greedy run by \Cref{lem:maximalcover}.
Furthermore, $r$ covers the non-greedy run as well since all of the elements of this run must already be in $A$'s initial buffer
and hence get written out. An optimal algorithm (with $M$-size buffer), on the unwritten-element-sequence, 
has to choose between the greedy and the non-greedy run. 
Since $A$ covers both choices of the optimal in one run, by \Cref{thm:template}, it is able to match or beat \OPT.
\end{proofsketch}

A natural question is whether resource augmentation boosts performance automatically, without using the run-simulation technique.
However, the following example shows that our 2-competitive algorithm,
even when allowed to have $4M$-size buffer, may still be as bad when using $M$-size buffer.

\begin{example}
Consider the input, $ (8M \searrow 1) \conc (16M \searrow 8M+1) \conc 
\cdots \conc (8cM \searrow 8(c-1)M+1) \ .$
The alternating algorithm
from \Cref{sec:two-approx}
which alternates maximal up and maximal down runs
will write $2c$ runs given a $4M$-size buffer. 
In contrast, the optimal number of runs with an $M$-size buffer has $c$ runs.
\end{example}

\pparagraph{$3/2$-competitive using $4M$-visibility} 
When we say that an algorithm has $X$-visibility ($X \geq M$) or $(X-M)$-lookahead, 
it means that the algorithm has knowledge of the next $X$ elements of its unwritten element sequence, and can use this knowledge when 
deciding what to write.  

However, only the usual $M$-size buffer is used for reading and writing.
Furthermore, the algorithm must continue to read elements into its buffer sequentially from $I$, 
even if it sees elements further down the stream it would like to read or rearrange instead.

We present a deterministic algorithm which uses $4M$-visibility to achieve a competitive ratio of $3/2$.
At each decision point, similar to the algorithm in \Cref{thm:optaug}, we can use $3M$-lookahead to determine
the direction leading to the longer (greedy) run.
However, unlike \Cref{thm:optaug} we cannot use a large buffer to write future elements.
Instead, we do the following---write a maximal greedy run, 
followed by two additional maximal runs in the same direction and opposite direction respectively.

We show that, at each decision point, the above algorithm is able to cover two runs of optimal (on the unwritten-element-sequence)
using three runs. \Cref{lem:pruning} and \Cref{lem:remaindersequence2} are key in this analysis (see \Cref{sec:omittedproofs} for details). 
Thus, we have the following.

\begin{theorem}
    \label{thm:visibility}
    Let $\OPT$ be the optimal number of runs on input $I$ given an $M$-size buffer, where $I$ has no duplicate elements. Then there exists an online algorithm $A$ with an $M$-size buffer and $4M$-visibility such that $A$ always outputs $S$ satisfying $R(S) \leq (3/2) \OPT$.
\end{theorem}

\pparagraph{$7/4$-competitive using $2M$-size buffer} 
We have seen that it is possible to achieve a competitive ratio of $3/2$ using a standard $M$-size buffer
as long as the algorithm is able to determine the direction leading to the longer (greedy) run (see \Cref{thm:visibility}).
Now we only have a $2M$-size buffer. 
The algorithm will pick a direction randomly, and write a maximal run in that direction using its regular $M$ buffer.
It use the additional $M$-size buffer to simulate a run in the opposite direction 
(and thus figure out which one is longer). 

With probability $1/2$, the algorithm is lucky and picks the greedy direction. In this case,
we can cover the first two runs of optimal (on the unwritten-element
sequence) with three runs as in \Cref{thm:visibility}. 
With probability $1/2$, the algorithm picks the wrong direction and
we spend four (alternating) runs to cover two runs of optimal.
Thus, in expectation we achieve a competitive ratio of $1/2(3/2)  + 1/2(4/2) = 7/4$.

\begin{theorem}\label{thm:sevebyfour}
 Let $\OPT$ be the optimal number of runs on input $I$ given an $M$-size buffer, where $I$ has no duplicate elements. Then there exists an online algorithm $A$ with a $2M$-size buffer such that $A$ always outputs $S$ satisfying $\expect [R(S)] \leq (7/4) \OPT$ and $R(S) \leq 2 \OPT$.
\end{theorem}

\begin{figure}
\centering
\ \begin{tabular} { | c | c | c | c |  c}
\hline
 Buffer size & Lookahead & Competitive ratio & Comments \\
    \hline
$M$ & - & 2 & Deterministic   \\
\hline
$2M$ & - & 1.75 & Randomized  \\ 
\hline
$M$ & $3M$ & 1.5  & Deterministic \\
\hline
$4M$ & - & 1 & Deterministic   \\
\hline
  \end{tabular}
  \caption{Summary of online algorithms on run generation on any input
    \label{fig:results}
}
\end{figure}
 
   \subsection{Lower Bound for Resource Augmentation}

We show that with less than $(4M-2)$-augmentation, no deterministic online
algorithm can be $3/2$-competitive on all inputs.  Thus, an algorithm with
$(4M-2)$-size buffer cannot be optimal, so \Cref{thm:optaug} is nearly tight.
Similarly, \Cref{thm:visibility} is nearly tight, since $4M-2$-size buffer
implies $4M-2$-visibility.

 \begin{theorem}\label{thm:res-aug-lower} With buffer size less than $(4M-2)$,
for any deterministic online algorithms $A$, there exists an input $I$ such
that if $S$ is the output of $A$ on $I$, then $R(S) \geq (3/2) \OPT$.
\end{theorem}

\section{Offline Algorithms for Run Generation}
\label{sec:offline}

We give offline algorithms for run generation. 
The offline problem is the following---given the entire input, 
compute (using a standard polynomial computation time algorithm) the optimal strategy
which when executed by a run generation algorithm (with a buffer of size $M$) produces the minimum possible number of runs.

For any $\epsilon$, we provide an offline polynomial time approximation algorithm that gives a $(1 + \epsilon)$-approximation to the optimal solution. This is called a polynomial-time approximation scheme, or PTAS. The running time of our first attempt is $O(2^{1/\epsilon}N\log N)$. We then improve the running time to $O(\varphi^{1/\epsilon}N\log N)$ where $\varphi = (1+\sqrt{5})/2 \approx 1.618$ is the well-known golden ratio.

\pparagraph{Simple PTAS}  Our first attempt breaks the output into sequences with a small number of runs, and uses brute force to find which set of runs writes the most elements.  We show that for any $\epsilon$, we can achieve a $1+\epsilon$ approximation in polynomial time using this strategy.

\begin{theorem} \label{thm:simplePTAS}
There exists an offline algorithm $A$ that always writes an $S$ satisfying $R(S) \leq (1+ \epsilon) \cdot \OPT$. The running time of $A$ is $O(2^{1/\epsilon}N\log N)$.
\end{theorem}

\pparagraph{Improved PTAS} 
We reduce the running time by bounding the number of choices we need to consider in a brute-force search. 
We do this using \Cref{lem:pruning}.

At each decision point, an algorithm chooses between starting an
increasing run $r_1$ and a decreasing run $r_2$. If
$|r_1| \geq |r_2|$, then by \Cref{lem:pruning},  
we are able to discard $r_2$ followed by an
increasing run.

	Let $F_d$ be the number of run sequences we need to consider if $d$
runs remain to be written (for example, na\"i{}ve PTAS has $F_d = 2^d$).
First, the algorithm must handle all run sequences beginning with $r_1$; this
is the same as an instance of $F_{d-1}$.  Then the algorithm handles all run
sequences beginning with $r_2$ followed by a decreasing run; this is an
instance of $F_{d-2}$.  Thus $F_d = F_{d-1} + F_{d-2}$; by examination, $F_1 =
1$ and $F_2 = 2$.  This is the Fibonacci sequence, which gives us the
$\varphi^{1/\epsilon}$ factor in the running time.

\begin{theorem}\label{thm:improvedPTAS} There exists an offline algorithm $A$
that writes $S$ such that $R(S) \leq (1+ \epsilon) \cdot \OPT$. The running
time of $A$ is $ O(\varphi^{1/\epsilon}N\log N)$ where $\varphi$ is the golden
ratio $ (1+\sqrt{5})/2$.
\end{theorem}


\section{Run Generation on Nearly Sorted Input}\label{sec:nearly-sorted} 

This section presents results proving 
that up-down replacement selection performs better when the
input has inherent sortedness (or ``bounded-disorder'' \cite{Martinez-PalauDoLa10}).
Replacement selection produces longer runs on
nearly sorted data. In particular, if every input element
is $M$ away from its target position, then a single run is produced. 
Similarly, we give algorithms which perform well on inputs, where the optimal runs are also long.

In particular, we say that an input is \defn{$c$-nearly-sorted} if there exists a proper optimal algorithm whose outputs consists of runs of length at least $cM$.

\pparagraph{$3/2$-competitive using $2M$-size Buffer} 
We provide a randomized online algorithm that, on inputs which are $3$-nearly-sorted, achieves a competitive ratio of $3/2$, while using
an augmented-buffer of size $2M$. 

A sketch of the algorithm follows. At each decision point, the algorithm picks a run
direction at random.  It starts a maximal run in that direction, but uses its extra $M$-buffer 
to simulate the run in the opposite direction. By \Cref{lem:3m}, the algorithm can tell if it picked the
same run as greedy (with $M$-buffer), similar to \Cref{thm:sevebyfour}.
If the algorithm got lucky and picked the greedy run, it repeats the process.  

If the algorithm picked the non-greedy run, 
it uses some careful bookkeeping to write elements and simulate the run in the opposite direction. 
In doing so, the algorithm winds up at the same point in the input it would have reached, 
had it written the greedy run in the first place, but with an additional cost of one run.

\begin{theorem}\label{thm:wellsorted1.5} 
There exists a randomized online algorithm $A$ using $M$ space in addition to its buffer such that, on
any 3-nearly-sorted input $I$ that has no duplicates, $A$ is a $3/2$-approximation in expectation. 
Furthermore, $A$ is at worst a
2-approximation regardless of its random choices.
\end{theorem}

\pparagraph{Exact Offline Algorithm on Nearly Sorted Input}
We show that the greedy (offline) algorithm is a linear time optimal algorithm 
on inputs which are $5$-nearly-sorted.
We first prove the following lemma.

\begin{lemma} \label{lem:5m}
If a proper algorithm produces runs of length at least $5M$ on a given input with no duplicates, then it is optimal. 
\end{lemma}

Thus, we get our required linear time exact offline algorithm.

\begin{theorem}
The greedy offline algorithm, i.e., picking the longer run at each decision point, is optimal on a $5$-nearly-sorted input that contain no duplicates.
The running time of the algorithm is $O(N)$. 
\end{theorem}

\section{Additional Related Work}\label{sec:related}

\pparagraph{Replacement Selection} The classic algorithm for run generation is
replacement selection \cite{Goetz63}.  While replacement selection considers
only up runs, Knuth \cite{Knuth63} analyzed alternating up-down replacement
selection in 1963.  He showed that for uniformly random input, alternating
up-down replacement selection produces runs of expected length $3M/2$, compared
to $2M$ of the standard replacement selection \cite{Knuth98, Friend56,
Gassner67}. 

Recently, Martinez-Palau et al. \cite{Martinez-PalauDoLa10} introduce {\em
Two-way replacement selection} (2WRS), reviving the idea of up-down replacement
selection. The 2WRS algorithm maintains two heaps in memory for up and down
runs and heuristically decides in which heap each element must be placed. Their
simulations show that 2WRS performs significantly better on inputs with mixed
up-down, alternating up-down, and random sequences.

Replacement selection with a fixed-sized {\em reservoir} appears in
\cite{FrazerWo72,TingWa77}.  Larson \cite{Larson03} introduced batched
replacement selection, a cache-conscious replacement selection which works for
variable-length records.  Koltsidas, M{\"u}ller, and Viglas
\cite{KoltsidasMuVi08} study replacement selection for sorting hierarchical
data.

Improvements for the merge phase of external sorting have been considered in
\cite{ZhengLa96, ZhangLa98, Salzberg89, Estivill-CastroWo94, ChandramouliGo14},
but this is beyond the scope of this paper.

\pparagraph{Reordering Buffer Management} Run generation problem is reminiscent
of the \emph{buffer reordering problem} (also known as the \emph{sorting buffer
problem}), introduced by R{\"a}cke et al.\ \cite{RackeSoWe02}.  It consists of
a sequence of $n$ elements that arrive over time, each having a certain color.
A buffer, that can store up to $k$ elements, is used to rearrange them.  When
the buffer becomes full, an element must be output. A cost is incurred every
time an element is output that has a color different from the previous element
in the output sequence.  The goal is to design a scheduling strategy for the
order in which elements must be output, so as to minimize the total number of
color changes. The buffer reordering problem models a number of important
problems in manufacturing processes and network routing and has been
extensively studied, both in the online and offline case \cite{ImMo14,
Avigdor-ElgrabliRaYu13, Avigdor-ElgrabliRa13-arXiv, Bar-YehudaLa07,
EnglertWe05, Avigdor-ElgrabliRa13, ChanMeSi11, AsahiroKa12}.  The offline
version of the buffer reordering problem is NP hard \cite{ChanMeSi11}, while
the complexity of our problem remains unresolved.

\pparagraph{Patience Sort and Longest Increasing Subsequence} An old sorting
technique used to sort decks of playing cards, {\em Patience Sort}
\cite{Mallows63} has two phases---the creation of sorted \emph{piles} or runs,
and the merging of these runs.  The elements arrive one at a time and each one
can be added to an existing run or starts a new run of its own.  Unlike this
paper, a legal run only consists of elements decreasing in value, and patience
sort can form any number of parallel runs.  The goal is to minimize the number
of runs.  The greedy strategy of placing an element to the left-most legal run
is optimal.  Moreover, the minimum number of such runs is the length of the
longest increasing subsequence of the input \cite{AldousDi99}.  Patience sort
has been studied in the streaming model \cite{GopalanJaKr07}.

Similar to Replacement Selection, Patience Sort
is able to leverage partially sorted input data.
Chandramouli and Goldstein \cite{ChandramouliGo14} present improvements to 
patience sort, and combine it with replacement selection to achieve practical speed up.

\pparagraph{Adaptive Sorting Algorithms} 
Python's inbuilt sorting algorithm,
{\em Timsort} \cite{wiki:timsort} works by finding contiguous runs of
increasing or decreasing value during the run generation phase.  
External memory sorting for well-ordered or ``partially sorted'' data has been studied
by Liu et al.\ \cite{LiuHeCh11}. They minimize the I/O cost of run generation
phase by finding ``naturally occurring runs''. See \cite{Estivill-CastroVlWo92}
for a survey on adaptive sorting algorithms.


\section{Conclusion and Open Problems}\label{sec:conclusion}
 
In this paper, we present an in-depth analysis of algorithms for run
generation.  We establish that considering both up and down runs can
substantially reduce the number of runs in an external sort.  The notion  of
up-down replacement selection has received relatively little attention since
Knuth's negative result~\cite{Knuth63}, until its promise was acknowledged by
the experimental work of Martinez-Palau et al.~\cite{Martinez-PalauDoLa10}.

The results in our paper complement the findings of 
Knuth~\cite{Knuth63} and
Martinez-Palau et
al.~\cite{Martinez-PalauDoLa10}.  In particular, strict up-down alternation
being the best possible strategy explains why heuristics for
up-down run-generation 
can lead to better performance in some cases. 
Moreover, our constant-factor
competitive ratios with resource augmentation and
lookahead may guide followup heuristics and practical speed-ups.

We conclude with open problems.

Can randomization help circumvent the lower bound of~$2$ on the competitive
ratio of online algorithms (without resource augmentation)?
We know that no randomized online algorithm can have a competitive ratio better
than $3/2$, but there is still a gap. 
What is the performance of the greedy offline algorithm compared to
optimal? We show that greedy can as bad as $3/2$ times optimal.
Is there a matching upper bound?  
Can we design a polynomial, exact, algorithm for the offline run-generation
problem? We find it intriguing that our attempts at an exact dynamic program 
requires maintaining too many buffer states 
to run in polynomial time.


\section{Acknowledgments}
\label{sec:ack}
We gratefully acknowledge 
Goetz Graefe and 
Harumi Kuno for introducing us to this problem and for 
their advice. 
This research was supported by
NSF grants
CCF~1114809, 
CCF~1217708,  
IIS~1247726, 
IIS~1251137,  
CNS~1408695,  
CCF~1439084,  
and by Sandia National Laboratories.

\bibliographystyle{plain}

\newpage
\appendix

\section{Appendix: Omitted proofs}\label{sec:omittedproofs}

\begin{proofof}{\Cref{lem:maximalcover}}
Without loss of generality, assume $r$ and $r'$ are increasing runs. 
Consider any time step $t$ when elements from both $r$ and $r'$ are being written. 
Let $B_t$ and $B_t'$ be the buffer of $r$ and $r'$ at time step $t$ respectively; let $C_t$ be the set of elements in $B_t$ that are eventually written to $r$, and $C_t'$ be the set of elements in $B'_t$ that are eventually written to $r'$.
We prove inductively that $C_t'\subseteq C_t$.
This implies that $r$ covers $r'$. 
The base case is true as $r$ and $r'$ start with the same buffer. Since we have $C'_t\subseteq C_t$,
we must show that (a) the element $z$ written to $r$ is not in $C'_{t+1}$ and 
(b) the element $e$ read into $C'_{t+1}$ must also be in $C_{t+1}$.  

Consider the elements $z'$ and $z$ written by $r'$ and $r$, respectively, at time $t+1$.  
We must have $z' \geq z$; thus either $z = z'$, or $z$ is never written to $r'$ (either way it is not in $C'_{t+1}$).

Since $e$ is in $C'_{t+1}$, it is eventually written by $r'$; thus $e \geq z'$.  Thus $e \geq z$, but that means $e$ is eventually written by $r$.  Since $e$ was just read it is in $B_{t+1}$; thus $e\in C_{t+1}$.
\end{proofof}

\begin{obsx}\label{obs:lower}
    If $A$ has just written an element $e$, and is writing a down (up) run, then $A$ cannot write any element larger (smaller) than $e$ in the same run.  Similarly, if $A$ has just written $e$, then $A$ cannot write both an element larger than $e$ and an element smaller than $e$ in the same run.
\end{obsx}

\begin{proofof}{\Cref{thm:2lower}}

Let $I_1$, the first $M$ elements of the input, be $I_1 = 1, 2, \ldots, M$.  
We divide the rest of the input into segments of size $M$.  Let the $(t+1)$st such segment be $I_{t+1}$.  Then,
$$I_{t+1} = (1 + tM \nearrow M + tM) \text{ or }
        (-(1 + tM)\searrow -(M + tM)) \ .$$

 Call this a positive segment and a negative segment respectively.  At time $M(t-1) + 1$ we decide whether $I_{t+1}$ is a positive or a negative segment based on $A$.   
 
 Specifically, we choose $I_{t+1}$ using either the direction of the run $A$ is writing, or the value of the most recent element written.  
 If $A$ is writing a down run, $I_{t+1}$ is a positive segment; if $A$ is writing an up run, $I_{t+1}$ is a negative segment.  It may be that $A$ has only written one element of a run (so $A$ could turn this into either an up run or a down run). If this element was the smallest element in the buffer of $A$, $I_{t+1}$ is a negative segment.  Otherwise, $I_{t+1}$ is a positive segment.

 First we show that $A$ must write at least one new run for each $I_{t}$; thus $R(S) \geq t$.  At least one run is required for $A$ to write $I_1$, so for the remainder of the proof we assume $t > 1$.  
 Consider time $M(t-2) + 1$, when $I_t$ begins.   We assume that $I_t$ is a positive segment---a mirroring argument works when $I_t$ is a negative segment.  Furthermore, note that the elements of $I_t$ are the largest in the instance so far.  
 
 There are two cases: $A$ is currently writing a down run, or the initial element of a new run.

Case 1: Algorithm $A$ is currently writing a down run.  Then the elements of $I_t$ must be larger than any element in $A$'s down run.  Thus $A$ must use another run to write the elements of $I_t$ by Observation \ref{obs:lower}.

     Case 2: Algorithm $A$ is writing the initial element of a new run.  By construction, the element written is not the smallest element in $A$'s buffer, but is smaller than all elements in $I_t$. 
     Then $A$ must spend one run to write the smallest element in its buffer, and another to write $I_t$.  Thus, $I_t$ causes $A$ to write a run in addition to its current run by Observation \ref{obs:lower}.

 On the other hand, an offline algorithm can write $I_{2i}$ and $I_{2i + 1}$ in one run.  Assume that $I_{2i}$ is a positive segment---a mirroring argument works when $I_{2i}$ is a negative segment.  If $I_{2i + 1}$ is positive, both can be written using an up run.  If $I_{2i+1}$ is negative, both can be written using a down run.  Thus $\OPT$ is no more than $\lceil t/2\rceil$.
\end{proofof}

\begin{proofof}{\Cref{thm:randlower}}
Our lower bound uses the same basic principles as \Cref{thm:2lower}.  We first show the lower bound with some repeated elements, then show how to perturb the elements to avoid repetitions.  We generate a randomized input and show that any deterministic algorithm cannot perform better than $3/2$ times $\OPT$ on that input.  The theorem is then proven by Yao's minimax principle. Yao's minimax principle  states that the best expected performance of a randomized algorithm is at least as large as the expected performance of the best \emph{deterministic} algorithm over a (known) distribution of inputs. (See, e.g.,  \cite{MotwaniRa10} for details.)

As in \Cref{thm:2lower}, we divide the input into segments of size $4M$.  Call these segments $I_t$ for $t = 1, 2, \ldots, \lfloor N/4M\rfloor$.  
Note that this input is randomized: for each $t$, we pick one of two inputs, each with probability $1/2$.
We choose either

\vspace{2pt}
\begin{tabular}{lcl}       
$I_t^1 = (1\nearrow M)$ & & $I_t^1 = (1\nearrow M)$ \\

$I_t^2 = (2M, \searrow M+1)$ & &  $I_t^2 = (-2M, \nearrow -M-1)$ \\

$I_t^3 = (3M\searrow 2M+1)$ & or & $I_t^3 = (-3M\nearrow -2M-1)$ \\

$I_t^4 = (4M\searrow 3M+1)$ & & $I_t^4 = (-4M\nearrow -3M-1).$\\

(a positive segment) & & \noindent (a negative segment)
\end{tabular}
\vspace{2pt}
 
Let $I_t = I_t^1\conc I_t^2\conc I_t^3\conc I_t^4$.  A positive or negative segment is chosen randomly for each $I_t$ with probability $1/2$.

The optimal algorithm spends no more than one run per $I_t$, using an up run for a positive segment or a down run for a negative segment. 

We show that any deterministic algorithm requires at least one new run to write $I_{t-1}^4$ and $I_{t}^1$ for $t > 1$. Further analysis shows that with probability $1/2$, any deterministic algorithm requires at least one run to write the remainder of $I_{t}^1$, $I_{t}^2$, and $I_t^3$.  Note that one run is also required to write $I_{1}^1$; summing, this gives a total expected cost of $(3/2)\OPT$.

Consider a segment $I_t$; $t > 1$.  Once all of $I_{t-1}^4$ has been read into its buffer, at least one element of $I_{t-1}^3$ has been written.  Once all of $I_{t}^1$ has been read into its buffer, at least one element of $I_{t-1}^4$ has been written.  Finally, once all of $I_{t}^2$ has been read into its buffer, at least one element $x$
of $I_t^1$ has been written.  Applying Observation \ref{obs:lower}, at least one new run is required to write these three elements.

Now we show that with probability $1/2$, an additional run is required to write $I_{t}^2$.  Let $x$ be the first element written by $I_{t}^1$ (thus, the cost of writing $x$ itself was handled in the above case---we show when an additional run is required).  
Note that the algorithm must choose an $x$ before it sees any element of $I_t^2$ (so it does now know if $I_t$ is positive or negative).

Let $x \neq 1$ and $I_t$ be a positive segment.
By Observation \ref{obs:lower}, an additional run is required to write both $1$ and any element of $I_t^2$.
If $1$ is not written, all of $I_t^2$ cannot be stored in the buffer---but then, $I_t^2$ and $I_t^3$ cannot be written using one run.
Similarly, let $x = 1$ and $I_t$ be a negative segment. 
By Observation \ref{obs:lower}, an additional run is required to write both $M$ and any element of $I_t^2$; 
otherwise $I_t^2$ and $I_t^3$ require an additional run to be written. 

Thus any deterministic algorithm cannot perform better than a $3/2$-approximation.   Applying Yao's minimax principle proves the theorem.

Now we perturb the input to avoid duplicate elements.  We multiply each element by $\lfloor N/4M\rfloor$, and add $t$ to each element of $I_t$.  In other words, we use a new segment 
$$I'_t = (I_t\otimes \lfloor N/4M\rfloor) \oplus t.$$ 
Our arguments above only depended on the relative ordering of the elements, which is preserved by this perturbation.  For example, assume $I_t$ and $I_{t-1}$ are both positive segments.  Then all elements of $I_{t}^1$ are less than all elements of $I_{t-1}^4$, and all elements of $I_t^2$ are greater than all elements of $I_{t}^1$.
\end{proofof}

\begin{proofof}{\Cref{thm:greedy_gaunratee}} 
We will build $S$ constructively.  At each time $t$ where greedy
chooses an up or down run, we show that one of its choices leads to a run of
length at least $5M/4$. Since the greedy algorithm always picks the longer
run at each decision time step, the run with length less than $5M/4$ can never be part of
its output.

Consider any time step $t$ where $t < N - 5M/4$.  If $t$ is larger than this
value, the final run will have length $n-t$. 
Let the contents of buffer at time $t$ be $\sigma_1, \ldots,
\sigma_M$.   Let $I'$
be the sequence
of $\lfloor M/2\rfloor$ elements of $I$ arriving after $t$.  

    Consider the run starting at $\sigma_M$ and continuing downwards; call this down run
$r_1$. Let $r_2$ be the up run starting at $\sigma_1$ and continuing upwards. 
Any element of $I'$ 
less than $\sigma_{\lceil M/2\rceil}$ will be (eventually) written out to $r_1$; any that are
greater than $\sigma_{\lfloor M/2\rfloor}$ will be written out to $r_2$.  Every
number must fall into one of these categories, so there must be at least
$\lceil\lfloor M/2\rfloor/2\rceil$ numbers added to the larger run.  Each run at least includes 
the elements already in the buffer, so the larger run has length at least $M + \lceil\lfloor M/2\rfloor/2\rceil$. 

The last $5M/4$ elements can be handled by greedy in at most 2 runs (since each trivially has length at least $M$).  Thus the above applies to all but the last two runs of greedy.
\end{proofof}

\begin{proofof}{\Cref{lem:3m}} 
Let $S_1$ be an output that writes $r_1$ initially and $S_2$ be an output that
writes $r_2$ initially.  Without loss of generality, suppose that $r_1$ is
increasing and $r_2$ is decreasing.  Let $r_1 = r_1(1),r_1(2), \ldots, r_1(k)$ and $r_2
=r_2(1), r_2(2), \ldots, r_2(\ell)$.  The idea of the proof is to split these runs into
two phases (a) elements of $r_1$ are smaller than the corresponding elements of
$r_2$ and (b) when the elements of $r_1$ are greater than or equal to those of
$r_2$.  During each of these two phases, we use the fact that incoming elements
written by $S_1$ have to be in the buffer of $S_2$ (and vice versa) to bound
their length.   We assume that both runs write exactly one element for each element they read in; this cannot affect the length of the runs.

Let $B_0$ be the original buffer, i.e., the first $M$ elements of $I$.  Let $i$
be the transition point between the two phases mentioned above; in other words,
$r_1(i+1) \geq r_2(i+1)$ but $r_1(i) < r_2(i)$.

Divide $r_1$ into $s_1$ and $t_1$, where $s_1$ is the first $i$ elements of $r_1$, and $t_1$ is the remainder of $r_1$.  
We further divide $s_1$ into $s_1^B$, the elements of $s_1$ that are in $B_0$, and $s_1^N$, the elements of $s_1$ that are not in $B_0$.  Let $t_1^B$ be the elements of $t_1$ that are in $B_0$.  Let $f_1$ be the set of elements in $r_1$ that are read in after $x_{i+1}$ is written.
Let $u_1$ be the set of elements not in $r_1$ that are read in before $x_{i+1}$ 
is written.  We define the corresponding sets for $r_2$ as well: $s_2^B, s_2^N, t_2^B, f_2, r_2,$ and $u_2$.

We can bound the size of several of these sets by $M$.
Note that $s_1$ cannot have more than $M$ elements, since all must be stored in the buffer while $s_2$ is being written.  Thus $|s_1^N| + |s_1^B| \leq M$. Similarly, $|s_2^N| + |s_2^B| \leq M$. We must also have $|u_1|\leq M$ and $|u_2|\leq M$.  Finally, consider $s_2^N \cup t_1^B$.  Any element in $s_2^N$ must be read before time step $i$. Since $s_1$ is disjoint from $s_2$ (by definition of $i$), all elements of $s_2^N$ must be in $S_1$'s buffer at time step $i$.  All elements of $t_1^B$ must also be in $S_1$'s buffer at time step $i$, so $|s_2^N| + |t_1^B|\leq M$.

Starting from time step $i+1$, any new element $e$ that is read in cannot be in both $r_1$ and $r_2$.  
This means that all elements of $f_1$ must be in the buffer of $S_2$ until $r_2$ ends, and all elements of $f_2$ must be in the buffer of $S_1$ until $r_1$ ends. 
On the other hand, all elements of $u_2$ must eventually be a part of $r_1$, and similarly for $u_1$ and $r_2$.

To begin, we show a weaker version of the lemma for runs of length $4M$.
We have $|r_1| \leq  
(|s_1^N| + |s_1^B|)  + 
(|u_2|) +
(|s_2^N|+ |t_1^B|) + 
|f_1| \leq 3M + |f_1|$. 
Then if $|r_1| \geq 4M$, then $|f_1| \geq M$.  Since all elements of $f_1$ must be stored in the buffer of $S_2$ until $r_2$ ends, $r_2$ must end when the $M$th element of $f_1$ is read in. Then
we must have $|r_2| < |r_1|$.

We have $|f_1| \geq M - |u_2|$; otherwise, 
$|r_1| \leq (|s_1^N| + |s_1^B|) + (|s_2^N| +  |t_1^B|) + (|u_2| + |f_1|) < 3M$. 
Consider the first $M - |u_2| - 1$ elements read in after $i$ that are eventually written to $r_1$ (this is a prefix of $f_1$), call them $f'_1$. 
Since $|f_1|\geq M - |u_2|$, there must be another element $e\in f_1$ that is read after all elements of $f_1'$.  Note $e\in r_1$.  Let $t$ be the time when $e$ arrives.

At $t$, the buffer of $S_2$ must contain all elements of $u_2$, as well as all elements of $f'_1$ and $e$.  
The buffer of $S_2$ is then full of elements that cannot be written in $r_2$.
Hence, $S_2$ is forced to start a new run at time $t < |r_1|$, so $|r_2| < |r_1|$.  Then we must have $|f_2| + |u_1| < M$, none of the elements in these sets are in $r_1$, and must be stored in $S_1$'s buffer until $r_1$ ends (which is after $r_2$ ends).
Finally, we have, 
\[|r_2| \leq (|s_2^N| + |s_2^B|) + (|s_1^N| + |t_2^B|) + (|u_1|  + |f_2|) < 3M \ ,\]
as required.

\end{proofof}

\begin{proofof}{\Cref{thm:optaug}}

Algorithm $A$ will simulate the maximal up run, $r_1$, and maximal down run, $r_2$, to see which is longer, but it does not actually need to write any elements during this simulation.
By \Cref{lem:3m}, if we find that one run has length at least $3M$, it must be the longer run.

We now describe exactly how to simulate a run $r$ using $4M$ space without writing any elements.
Algorithm $A$ simulates the run one step at a time.  We describe the actions and buffer of an algorithm with $M$-size buffer writing $r$ as the \defn{simulated algorithm}.
Without loss of generality, assume $r$ is an up run.

Assume that all elements are stored in the buffer in the order they arrive.  
Thus, after $t$ elements have been written, 
the first $M+t$ elements of the buffer are exactly the elements the simulated algorithm has read from the input up to time $t$.
Of these elements, $M$ must be in the buffer of the simulated algorithm, while the other $t$ will have been written to $r$; however, $A$ does not explicitly keep track of which elements are in the buffer.

The algorithm $A$ keeps track of $\ell$, the last element written to $r$, 
because at each $t$, all of the first $M+t$ elements larger than $\ell$ are: (a) in the simulated algorithm's buffer at time $t$ and (b) will be written to $r$ at a later point.  Thus, once no item in the first $M+t$ elements is larger than $\ell$, $r$ must end.
At each time step, the smallest element larger than $\ell$ is written to $r$.  

Specifically, at time step $t$, $A$ finds the smallest element $e$ in the first $M+t$ elements of the buffer that is larger than $\ell$. 
This is the next element of $r$. 
Thus in the next time step, $A$ updates $\ell \gets e$, and repeats.  If no such $e$ can be found, no element in $M+t$ (and thus no element in the buffer) can continue the run, so the run ends at time $t$.  

The last time $A$ can update $\ell$ is when the simulated algorithm has seen all elements in the buffer; in other words, $t = 3M$.  By \Cref{lem:3m}, this is sufficient to determine which run is longer.

The algorithm now knows which run is longer; without loss of generality, assume $|r_1| \geq |r_2|$. Then the algorithm writes a maximal run $r$ using its $4M$-size buffer in the direction of $r_1$. 
Run $r$ is guaranteed to contain all elements of $r_1$ by \Cref{lem:maximalcover}.  Since $r_2$ has length less than $3M$ by \Cref{lem:3m}, all of its elements must already be in the $4M$-size buffer.  Thus they are written during $r$ because a maximal run always writes its buffer contents. The first initial run of a proper optimal algorithm on the unwritten-element sequence has to be either $r_1$ or $r_2$. Since $r$ covers both $r_1$ and $r_2$, by \Cref{thm:template} with $\beta =1$, $A$ never writes more runs than an optimal algorithm with an $M$-size buffer.
\end{proofof}

\begin{lemx} \label{lem:comparable_buffer}
Consider two algorithms $A_1$ and $A_2$ that have the same remaining input $I$ when they both start writing a new maximal run, called $r_1$ and $r_2$. Let their buffers at this point be $B_1$ and $B_2$ (that may not be full) respectively, and assume $\max(B_1 \setminus B_2) \leq \min (B_2 \setminus B_1)$.
If $r_1$ and $r_2$ are increasing 
then all elements in $I$ written to $r_2$ are also written to $r_1$.  Similarly, if $r_1$ and $r_2$ are decreasing, all elements in $I$ written to $r_1$ are also written to $r_2$.
\end{lemx}

\begin{proof}
It suffices to prove the first case where $r_1$ and $r_2$ are both increasing as the other case can be proven similarly. After $r_1(t)$ and $r_2(t)$ were written, let $C_1(t)$ and $C_2(t)$ be the set of elements in the buffers of $A_1$ and $A_2$ that will be written in $r_1$ and $r_2$ at some point in the future, i.e., the set of elements that are at least as large as $r_1(t)$ or $r_2(t)$ respectively.

It is easy to prove by induction that the invariant $$\max(C_1(t) \setminus C_2(t)) \leq \min (C_2 (t)\setminus C_1(t))$$ always holds. We note that this invariant implies $r_1(t+1) \leq r_2(t+1)$. Therefore, if this invariant is true for all $t < \min \{ |r_1| , |r_2| \}$, any incoming element $e \in I$ satisfies $e \in r_2 \Rightarrow e \in r_1$ as required. We now prove the invariant:

The base case is true since $C_1(0) = B_1, C_2(0)=B_2$. Suppose the invariant holds for $t$, then $r_1(t+1) \leq r_2(t+1)$ and a new element $e$ is read in.   

Case 1: if $e \in C_2(t+1) \Rightarrow e \geq r_2(t+1) \geq r_1(t+1) \Rightarrow e \in C_1(t+1)$. 

Case 2: if $e \in C_1(t+1)$ and $e \notin C_2(t+1)$, then $e < r_2(t) = \min(C_2(t)) \leq \min(C_2(t+1))$.

Hence, the invariant holds for $t+1$.

\end{proof}

\begin{obsx}
\label{obs:optimalprefix} 
On an input $I$, let $r_1 \conc \ldots \conc r_k$ be the first $k$ runs of an optimal output.
If $r'_1 \conc \ldots \conc r'_k$ be $k$ runs that cover $r_1 \conc \ldots \conc r_k$. Then,
$r'_1 \conc \ldots \conc r'_k$ are also the first $k$ runs of an optimal output.
\end{obsx}

\begin{proofof}{\Cref{lem:pruning}}

Without loss of generality, assume $r_1$ and $r_2$ are initial maximal increasing and decreasing runs respectively and $|r_1| \geq |r_2|$. Suppose $r_2 \conc r_3$, where $r_3$ is an increasing, is prefix of an optimal output $S_{{\rm\sc OPT}}(I)$. 

Consider the case $|r_1|=|r_2|$. 
Let their buffers at the end of these two runs be $B_1,B_2$ and let $j$ be the smallest index such that $r_1(j+1) > r_2(j+1)$. Consider any new element $e \in I$ that is read in before $r_1(j+1)$ is written. Obviously, we have:
\begin{align*}
e \in B_1  \Rightarrow e \leq r_1(j), \\
e \in B_2 \Rightarrow e \geq r_2(j). 
\end{align*} 
Consider any new element $e$ which is read in after $r_1(j+1), r_2(j+1)$ were written. It is easy to see the followings:
\begin{align*}
 e \in B_1 \setminus B_2 \Rightarrow e \leq r_1(j+1), \\
 e \in B_2 \setminus B_1 \Rightarrow e \geq r_2(j+1) .
 \end{align*}Therefore, $\max (B_1 \setminus B_2) \leq \max(r_1(j),r_2(j+1)) \leq \min(r_2(j),r_1(j+1)) \leq \min (B_2 \setminus B_1)$. The situation can be visualized in \Cref{fig:comparable} as follows. If an incoming element cannot be written in the current run, it lies below or above (depending on whether the run is increasing or decreasing) the last element written. The regions are marked with their associated sets described above.

\begin{figure}[H]
\centering
\includegraphics[scale=0.5]{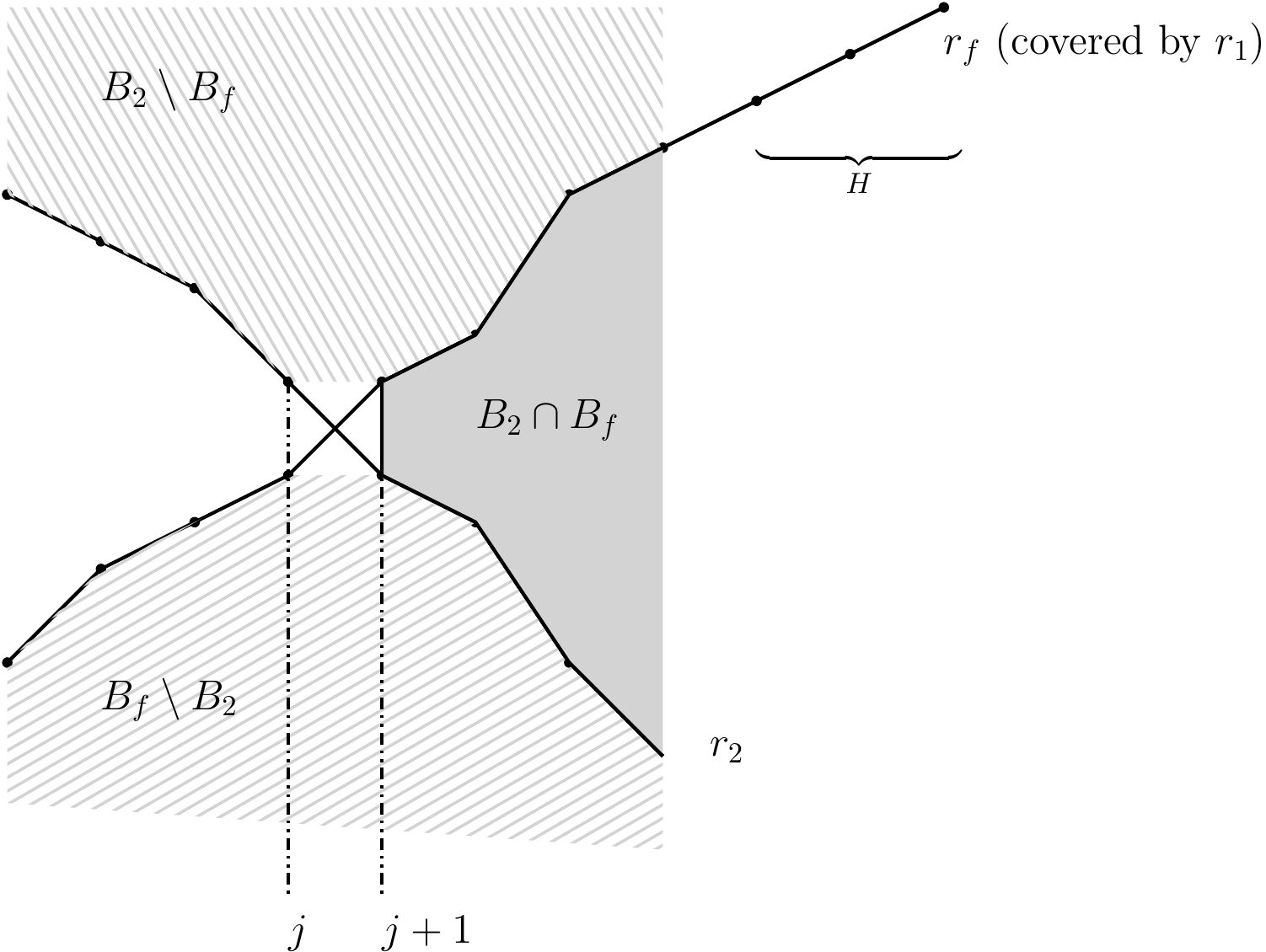}
\caption{Visualizing the buffer states.}
 \label{fig:comparable}
\end{figure}

Consider $r_1 \conc r_4$, where $r_4$ is a maximal increasing run. Every elements in $r_2$ will be written in either $r_1$ or $r_4$ by \Cref{lem:remaindersequence2}. If $e \in r_3$, then we consider the cases where $e \in B_2$ or $e \in r_3 \setminus B_2$. If $e \in B_2$, then $e$ is either in $r_1$ or in $B_1$ which means $e$ is either in $r_1$ or $r_4$. If $e \in r_3 \setminus B_2$, $e \in r_4$ by Lemma \ref{lem:comparable_buffer} using the fact that $\max (B_1 \setminus B_2) \leq \min (B_2 \setminus B_1)$. Thus, $r_1 \conc r_4$ covers $r_2 \conc r_3$. As a result, $r_1$ is also a prefix of an optimal output $S_{{\rm\sc OPT}}(I)'$ by Observation \ref{obs:optimalprefix}.

If $|r_1| > |r_2|$ and $|r_2| = k$. Then the simplest argument goes as follows. Instead of arguing based on $r_1$ directly, we consider $r_f$ that is increasing but may not  be maximal. Consider an algorithm $A_{f}$ that writes $r_f(1) = r_1(1),\ldots,r_f(k) = r_1(k)$. Then, without reading any new element in, it finishes its first run $r_f$ by writing out all elements in its buffer that are larger than $r_1(k)$ (the set of these elements is $H$ in \Cref{fig:comparable}). After this extra step, let the buffer be $B_{f}$. Use the exact same argument as above, we have that $\max (B_{f} \setminus B_2)  \leq \max(r_f(j),r_2(j+1)) \leq \min(r_2(j),r_f(j+1)) \leq \min (B_2 \setminus B_{f})$. Using the same argument as in the first case, we have that $r_f$ followed by a maximal increasing run will cover $r_2$ followed by a maximal increasing run. Hence, $r_f$ is a prefix of an optimal output. Since $r_1$ covers $r_f$, it is also a prefix of an optimal output by Observation \ref{obs:optimalprefix} . 
\end{proofof}

\begin{proofof}{\Cref{thm:visibility}}
     In any partition that is not the last one, let $I_r$ be the unwritten-element sequence and 
     let $r_1,r_2$ be the two possible maximal initial runs where $r_1$ is increasing and $r_2$ is decreasing. Without loss of generality, suppose $|r_1| \geq |r_2|$. 
     We use the simulation technique of \Cref{thm:optaug} to determine which run is longer.
         
     The algorithm writes $r_1 \conc r_3 \conc r_4$ where $r_3$ and $r_4$ are maximal runs that have the same and opposite directions as $r_1$ respectively. The algorithm stops when there is no element left to write. We break our analysis into cases based on what runs are in an optimal output.  In each case, we show that \Cref{thm:template} proves a competitive ratio of $\beta = 3/2$.
     
     If $r_2$ is a prefix of a proper optimal output $S_{{\rm\sc OPT}}(I_r)$,
     let $r_5$ be the maximal run after $r_2$ in $S_{{\rm\sc OPT}}(I_r)$.
     After writing $r_3$, the algorithm already writes out all elements in $r_2$ by \Cref{lem:remaindersequence2}. Let the unwritten-element sequence after writing $r_3$ be $I_3$ and let the unwritten-element sequence after writing $r_2$ be $I_2$. By \Cref{lem:remaindersequence2}, $I_3$ is a subsequence of $I_2$. According to \Cref{lem:pruning}, $r_5$ has to be decreasing in order to possibly have fewer runs than writing $r_1$ initially. Hence, applying  \Cref{lem:subsequence} to $r_4$ and $r_5$, we know that at the end of $r_4$, the algorithm has written all elements of $r_2$ and $r_5$. Thus, $r_1 \conc r_3 \conc r_4$ 
     covers $r_2 \circ r_5$.

If $r_1 \conc r_3$ is a prefix of $S_{{\rm\sc OPT}}(I_R)$, then we are done as $r_1 \conc r_3 \conc r_4$ trivially covers $r_1\conc r_3$.

If $r_1$ is a prefix of $S_{{\rm\sc OPT}}(I_R)$ but $r_1 \circ r_3$ is not a prefix of $S_{{\rm\sc OPT}}(I_R)$. Then, let $r_3'$ be the opposite maximal run to $r_3$, i.e., $r_1, r_3'$ are the first two runs of $S_{{\rm\sc OPT}}(I_R)$. We have $I_3$ is a subsequence of $I_1$. Hence, applying \Cref{lem:subsequence} to $r_3'$ on input $I_1$ and $r_4$ on input $I_3$, we have that at the end of $r_4$, the algorithm has written out all elements in $r_1 \conc r_3'$. Thus, $r_1 \conc r_3 \conc r_4$ 
covers $r_1 \circ r_3'$.

In the last partition, since $A$ outputs at most $3$ runs, it can only achieve a ratio worse than $3/2$ if the optimal algorithm wrote out a single run.  But then that run is longer, and $A$ would choose it.
Therefore, we have $R(S) \leq (3/2) \cdot \OPT$. 
\end{proofof}

\begin{proofof}{\Cref{thm:sevebyfour}}
In each partition, let the unwritten-element sequence be $I_r$ and the optimal proper output of $I_r$ be $S_{\OPT}(I_r)$. The algorithm randomly picks the direction of the next run and writes a maximal runs in that direction using $M$-size buffer. It uses the extra $M$ buffer slots to simulate the buffer state of the maximal run in the other direction to check if the run it chose is at least as long as the other run. If the algorithm picked the run that is at least as long as the other run, it then writes a maximal run in the same direction followed by another maximal run in the opposite direction. The algorithm stops when there is no more element to write. In the proof of \Cref{thm:visibility}, we showed that these three runs will cover the first two runs of $S_{\OPT}(I_r)$.

If the algorithm picked the shorter run, then it writes three more maximal runs with alternating directions. We know that the first two runs with alternating directions cover the first run of $S_{\OPT}(I_r)$ as argued in the proof of \Cref{thm:2competitive}; hence, the next two runs with alternating directions cover the second run of $S_{\OPT}(I_r)$. 

In the last partition, if $\OPT(I_r) = 2$, the analysis is the same. If $\OPT(I_r) = 1$, then the optimal output must be the longer maximal run. The algorithm, if picked the shorter run, then will cover the longer run when it writes the next maximal run in the opposite direction as showed in the proof of \Cref{thm:2competitive}. Therefore, we have
${\expect [x_i]}/{y_i} \leq {1}/{2} . ({4}/{2}) + {1}/{2} . ({3}/{2}) =  {7}/{4}.$


Applying \Cref{thm:template} with $\alpha = 1.75$ and $\beta = 2$, we have: $\expect [R(S)] \leq (7/4)\OPT$ and $R(S) = 2\OPT$.
\end{proofof}

\begin{proofof}{\Cref{thm:res-aug-lower}}
Suppose an algorithm has $(4M-3)$-size buffer. Consider the input $I_1 \conc e \conc I_2$ where $I_1 = (1 \nearrow M-1) \conc (2M-1 \searrow  M) \conc (3M \nearrow 4M-2) \conc (-M \searrow -2M+2)$. 

If $S$ first writes $-2M+2$, then let $e = -2M+1$. 
\begin{itemize}[noitemsep,nolistsep,leftmargin=*]
\item Case 1: if $S$ writes $e = -2M+1$ next, then let $I_2 = (0 \searrow -(M-1))$.  Thus, $S$ has to spend at least two runs while an optimal output is one run:
$(4M-2 \searrow 3M) \conc (2M-1 \searrow -2M+1)$. 
\item Case 2: if $S$ writes $-2M+3$ next, let $I_2 = (-2M \searrow -10M) \conc (2M \nearrow 3M-1)$. Then $S$ has to spend at least $3$ runs while an optimal output has $2$ runs: $(4M-2 \searrow 3M) \conc (2M-2 \searrow -10M) \conc (2M \nearrow 3M-1)$. 
\end{itemize}
Similarly, if $S$ first writes $4M-2$, then let $e = 4M-1$. 
\begin{itemize}[noitemsep,nolistsep,leftmargin=*]
\item Case 1: if $S$ writes $e=4M-1$ next, then let $I_2 = (2M \nearrow 3M-1)$. 
\item Case 2: if $S$ writes $4M-3$ next, let $I_2 = (4M \nearrow 10M) \conc (0 \searrow -M+1)$.
\end{itemize}
If $S$ first writes $e' \notin \{-2M+2,4M-2\}$, then let $e= -2M+1, I_2 = (0 \searrow -(M-1))$.  
Thus, $S$ has to spend at least two runs while an optimal output has the 
following output with one run:
$(4M-2 \searrow 3M) \conc (2M-1 \searrow -2M+1).$
\end{proofof}

\begin{proofof}{\Cref{thm:simplePTAS}}
We apply \Cref{thm:template} with $x = (\lceil 1/\epsilon \rceil + 1)$ and $y = \lceil 1/\epsilon \rceil$. In any partition except the last one, the algorithm chooses the combination of $\lceil 1/\epsilon \rceil$ maximal runs $r_1 \conc \cdots \conc r_{\lceil 1/\epsilon \rceil}$ whose output is longest (ties are broken arbitrarily) and writes out one extra run $r_{(\lceil 1/\epsilon \rceil + 1)}$. By \Cref{lem:remaindersequence2}, $r_1 \conc \cdots \conc r_{\lceil 1/\epsilon\rceil+1}$ covers the first $\lceil 1/\epsilon \rceil$ runs of 
a proper optimal output of the unwritten-element sequence $I_r$
in $(1 + \lceil 1/\epsilon \rceil)$ runs. In the last partition, the algorithm chooses a combination of runs with the smallest number of runs.

Therefore, we 
obtain an   
$\beta = 1 + {1}/{\lceil 1/\epsilon\rceil} \leq 1 + \epsilon $ approximation.

There are $2^{\lceil 1/\epsilon\rceil + 1}$ combinations to consider (each run can be up or down). 
The length of a run can be calculated in $O(N_i)$ time by simulating it directly.
where $N_i$ is the length of the longest output, namely, $|r_1\conc\cdots\conc r_{\lceil 1/\epsilon\rceil+1}|$, 
Since $N_i$ items are then written out, the total running time is $O(\sum_{i=1}^{t} N_i 2^{\lceil 1/\epsilon\rceil}) = O(N2^{1/\epsilon})$.   Searching for the shortest way to write out the remaining elements (once $I_r = \emptyset$) takes $O(N2^{1/\epsilon})$ time, which does not affect the running time.
\end{proofof}

\begin{proofof}{\Cref{thm:improvedPTAS}}
    In each partition, we restrict the search for the combination of $(1/\lceil \epsilon \rceil)$ consecutive runs that writes the longest sequence as described above. 
    By \Cref{lem:pruning}, if $d$ runs remain to be written out, we must examine one subcase with $d-1$ runs remaining, and one with $d-2$ runs remaining.  Thus, the number of combinations we need to consider is $F_d = F_{d-1} + F_{d-2}$.  
    Therefore, the running time of this step is  $O(F_{\lceil 1/\epsilon\rceil} N_i\log N_i)$. Thus, we have 
    \[O(F_{\lceil 1/\epsilon\rceil} N_i\log N_i) = O(\varphi^{\lceil 1/\epsilon\rceil} N_i\log N_i).\] This is because $F_{\lceil 1/\epsilon\rceil}  = (\varphi^{\lceil 1/\epsilon\rceil} - \psi^{\lceil 1/\epsilon\rceil})/\sqrt{5} \leq \varphi^{\lceil 1/\epsilon\rceil}.$
\end{proofof}

\begin{proofof}{\Cref{thm:wellsorted1.5}}
At each decision time step, $A$ flips a coin to pick a direction for the next run $r$.  It begins writing an up or down run according to the coin flip.  

Meanwhile, $A$ uses $M$ additional space to simulate $r'$, the run in the opposite direction.  In particular, it simulates the contents of the buffer at each time step, as well as the last element written.  Note that $A$ does not need to keep track of the most recent element read when simulating $r'$, as it is always the last element in the buffer.

By \Cref{lem:3m}, the run with the incorrect direction has length less than $3M$ and the run with the correct direction has length of $3M$ or more. Thus, $A$ can tell if it picked the correct direction. With probability $1/2$, $A$ writes the longer run.  
Therefore, it knows it made the correct direction and repeats, flipping another coin.
 
Now consider the case where $A$ picks the wrong direction.  When $r$ ends (at time $t$), $r'$ is continuing.  Then $A$ must act \emph{exactly} as if it had written $r'$.  Specifically, we cannot simply cover $r'$ and use an argument akin to \Cref{thm:template}, as then the unwritten-element sequence may not be 3-nearly-sorted.  

To simulate $r'$, $A$ has two tasks: (a) $A$ must write all elements that were written by $r'$ that were not written by $r$, and (b) $A$ must ``undo'' writing any element that was written during $r$ that is not in $r'$, in case these elements are required to make a subsequent run have length $3M$.  Divide the buffer into two halves: $B_r$ is the buffer after writing $r$, and $B_{r'}$ is the buffer being simulated when $r$ ends.

The first task is to ensure that $A$ writes all elements written by $r'$ that were not written by $r$ in the direction of $r'$.  These elements must be in $B_r$ since they were not written by $r$; and they must not be in $B_{r'}$ since they were written by $r'$.  Thus, $A$ can simply write out each element in $B_r$ that is not in $B_{r'}$ and continue writing $r'$ from that time step.

The second task is to ensure that all elements written during $r$ that were not written during $r'$ cannot affect future run lengths.  These elements must be in $B_{r'}$ but not in $B_r$.  We mark these elements as special ghost elements.  We can do this with $O(1)$ additional space by moving them to the front of the buffer and keeping track of how many of them there are.  During subsequent runs, these are considered to be a part of $A$'s buffer.  However, when $A$ would normally want to write one of these elements out, it instead simply deletes it from its buffer without writing any element.  That said, $A$ still counts these deletions towards the size of the run.
Note that our buffer never overflows, as $A$ continues to write (or delete) one element per time step.  

When this simulation is finished, the contents of $A$'s buffer are exactly what they would have been had it written out $r'$ in the first place---however some are ghost elements, and will be deleted instead of written.  Then $A$ repeats, flipping another coin.

For each run in the optimal output, either: $A$ writes that run exactly for cost 1 (with probability $1/2$), or $A$ writes another run, and makes up for its mistake by simulating the correct run exactly, for cost 2 (with probability $1/2$).  Thus $A$ has expected cost $(3/2)\OPT$.  In the worst case, $A$ guesses incorrectly each time for a total cost of $2\OPT$.

\end{proofof}

\begin{proofof}{\Cref{lem:5m}}

Suppose we are at a decision time step. Without loss of generality, assume this time step to be 0. Let the next two possible maximal runs be $r_1$ and $r_2$ that are up and down respectively. Without loss of generality, suppose $|r_1| \geq 5M$. Let $r_3$ be the maximal decreasing run that follows $r_2$. By \Cref{lem:pruning}, either writing $r_1$ or writing $r_2 \conc r_3$ is optimal on the unwritten-element sequence. Call the two outputs $S_1$ and $S_2$ respectively. 

Let $s_1^B,s_1^N,t_1^B,f_1,u_1,s_2^B,s_2^N,t_2^B,f_2,u_2$ be the same sets described in the proof of \Cref{lem:3m}. Let $r_1 = (x_1,\ldots,x_k),r_2= (y_1,\ldots,y_\ell),r_3 = (y_{\ell+1},\ldots,y_q)$. Let the buffers of $S_1$ and $S_2$ after time step $t+\ell$ be $B_{1,\ell},B_{2,\ell}$. Let $j \geq \ell$ be the smallest such that $x_{j+1} \geq y_{j+1}$. 

Similar to the proof of \Cref{lem:3m}, we let $r_{1,2} = \{x_{\ell+1},\ldots,x_{j} \}$ and $t_{1,2} = \{x_{j+1},\ldots,x_k \}$. Let $s_{1,2}^N$ be the set of elements in $r_{1,2}$ but not in $B_{1,\ell}$ . Let $s_{1,2}^B $ be the set of elements in $r_{1,2}$ and also in $B_{1,\ell} \setminus r_2$. Let $t_{1,2}^B$ be the set of elements in $t_{1,2}$ and also in $B_{1,\ell} \setminus r_2$. Let $u_{1,2}$ be the set of elements not in $r_1$ and read in before $x_{j+1}$ is written. Let $f_{1,2}$ be the set of elements in $r_1$ and read in after $x_{j+1}$ is written.

We define $s_3 = \{y_{\ell+1},\ldots,y_{j} \}$ and $t_3 =  \{ y_{j+1},\ldots,y_q \}$. Let $s_{3}^N $ be the set of elements in $s_3$ but not in $B_{2,\ell}$. Let $s_{3}^B $ be the set of elements in $s_3 \cap B_{2,\ell}$. Let $u_{3}$ be the set of elements that are not in $r_3$ and read in before $y_{j+1}$ is written. Let $f_{3}$ be the set of elements in $r_3$ and read in after $y_{j+1}$ is written.

Since the buffer of $S_2$ must keep all elements in $s_{1,2}^N , s_{1,2}^B$ at time step $j+1$, we have $|s_{1,2}^N |+ |s_{1,2}^B| \leq M$. We have, 
\begin{align*}
|r_1| & = |r_2| +  (|u_3|+|s_3^N|) + (|s_{1,2}^N| + |s_{1,2}^B|) + |f_{1,2}| \\
 & \leq (2M + |u_1| +|f_2|) + (M-|u_1| -|f_2| )+  M +|f_{1,2}| \\
 & = 4M + |f_{1,2}|.
\end{align*}
Since  $|f_{1,2}| \geq M$ because of our assumption $|r_1| \geq 5M$, $r_3$ has to end before $r_1$ using the same argument as in the proof of \Cref{lem:3m}. Since $|r_1| \geq |r_2 \conc r_3|$, $r_1$ followed by any maximal run will cover $r_2 \conc r_3$. Therefore, $r_1$ is an optimal prefix of the unwritten-element sequence. At every time step, the maximal run of length $5M$ or more is always a prefix of an optimal output on the unwritten-element sequence as required.
\end{proofof}

\end{document}